\documentclass[proceedings]{stacs}
\stacsheading{2010}{537-548}{Nancy, France}
\firstpageno{537}

\usepackage{ifpdf}
\ifpdf
\usepackage{pdfsync}
\fi
\usepackage{amsmath,amssymb,enumerate}

\newcommand{\bN}{{\mathbb N}}
\newcommand{\ramsey}{\to}
\newcommand{\lex}{{\mathrm{lex}}}
\newcommand{\pref}{{\mathrm{pref}}}
\newcommand{\supp}{{\mathrm{supp}}}
\newcommand{\alephN}{{\aleph_0}}
\newcommand{\cont}{{2^\alephN}}
\newcommand{\oCF}{\omega\mathrm{CF}}
\newcommand{\ocCF}{\mathrm{co}\text{-}\omega\mathrm{CF}}
\newcommand{\erCF}{\omega\mathrm{erCF}}
\newcommand{\cerCF}{\mathrm{co}\text{-}\omega\mathrm{erCF}}
\newcommand{\LANG}{\mathrm{LANG}}
\newcommand{\oLANG}{\omega\mathrm{LANG}}
\newcommand{\REG}{\mathrm{REG}}
\newcommand{\oREG}{\omega\mathrm{REG}}

\begin{document}

% \title[short title]{title}
\title[Is Ramsey's theorem $\omega$-automatic?]{Is Ramsey's theorem $\omega$-automatic?}

\author{Dietrich Kuske}{Dietrich Kuske} 

\address{Centre national de la recherche scientifique (CNRS)
  and\newline Laboratoire Bordelais de Recherche en Informatique
  (LaBRI), Bordeaux, France}
\thanks{These results were obtained when the author was affiliated
  with the Universit\"at Leipzig.} %optional

%% mandatory lists of keywords and classifications:
% \titlecomment{OPTIONAL comment concerning the title, \eg, if a variant
% or an extended abstract of the paper has appeared elsewehere}
%%%%%%%%%%%%%%%%%%%%%%%%%%%%%%%%%%%%%%%%%%%%%%%%%%%%%%%%%%%%%%%%%%%%%%%%%%%

%% the abstract has to PRECEDE the command \maketitle:
%% be sure not to issue the \maketitle command twice!

\begin{abstract}
  \noindent We study the existence of infinite cliques in
  $\omega$-automatic (hyper-)graphs. It turns out that the situation
  is much nicer than in general uncountable graphs, but not as nice as
  for automatic graphs.

  More specifically, we show that every uncountable $\omega$-automatic
  graph contains an uncountable co-context-free clique or anticlique,
  but not necessarily a context-free (let alone regular) clique or
  anticlique. We also show that uncountable $\omega$-automatic ternary
  hypergraphs need not have uncountable cliques or anticliques at all.
\end{abstract}

\maketitle

%% start the paper here:
\section*{Introduction}\label{S:one}
Every infinite graph has an infinite clique or an infinite anticlique
-- this is the paradigmatic formulation of Ramsey's theorem
\cite{Ram30}. But this theorem is highly non-construc\-tive since there
are recursive infinite graphs whose infinite cliques and anticliques
are all non-recursive (not even in $\Sigma^0_2$, \cite{Joc72}, cf.\
\cite[Thm.~4.6]{Gas98}). Recall that a graph is recursive if both its
set of nodes and its set of edges can be decided by a Turing
machine. Replacing these Turing machines by finite automata, one
obtains the more restrictive notion of an \emph{automatic graph}: the
set of nodes is a regular set and whether a pair of nodes forms an
edge can be decided by a synchronous two-tape automaton (this concept
is known since the beginning of automata theory, a systematic study
started with \cite{KhoN95,BluG04}, see \cite{Rub08} for a recent
overview). In this context, the situation is much more favourable:
every infinite automatic graph contains an infinite regular clique or
an infinite regular anticlique (cf.~\cite{Rub08}).

Soon after Ramsey's paper from 1930, authors got interested in a
quantitative analysis. For finite graphs, one can ask for the minimal
number of nodes that guarantee the existence of a clique or anticlique
of some prescribed size. This also makes sense in the infinite: how
many nodes are necessary and sufficient to obtain a clique or
anticlique of size $\alephN$ (Ramsey's theorem tells us: $\alephN$) or
$\aleph_1$ (here one needs more than $\cont$
nodes~\cite{Sie33,ErdR56}).

Since automatic graphs contain at most $\alephN$ nodes, we need a more
general notion for a recursion-theoretic analysis of this
situation. For this, we use Blumensath \& Gr\"adel's \cite{BluG04}
$\omega$-automatic graphs: the names of nodes form a regular
$\omega$-language and the edge relation (on names) as well as the
relation ``these two names denote the same node'' can be decided by a
synchronous 2-tape B\"uchi-automaton. In this paper, we answer the
question whether these $\omega$-automatic graphs are more like
automatic graphs (i.e., large cliques or anticliques with nice
properties exist) or like general graphs (large cliques need not
exist).

Our answer to this question is a clear ``somewhere in between'': We
show that every $\omega$-automatic graph of size $\cont$ contains a
clique or anticlique of size $\cont$ (Theorem~\ref{T-uncountable}) --
this is in contrast to the case of arbitrary graphs where such a
subgraph need not exist~\cite{Sie33}. But in general, there is no
regular clique or anticlique (Theorem~\ref{T-complexity}) -- this is in
contrast with the case of automatic graphs where we always find a
large regular clique or anticlique. Finally, we also provide an
$\omega$-automatic ``ternary hypergraph'' of size $\cont$ without any
clique or anticlique of size~$\aleph_1$, let alone~$\cont$
(Theorem~\ref{T-S2}).

For Theorem~\ref{T-uncountable}, we re-use the proof
from~\cite{BarKR08} that was originally constructed to deal with
infinity quantifiers in $\omega$-automatic structures. The proof of
Theorem~\ref{T-complexity} makes use of the ``ultimately equal''
relation. This relation was also crucial in the separation of
injectively from general $\omega$-automatic structures~\cite{HjoKMN08}
as well as in the handling of infinity quantifiers in \cite{KusL08}
and \cite{BarKR08}. In the ternary hypergraph from Theorem~\ref{T-S2},
a 3-set $\{x,y,z\}$ of infinite words with $x<_\lex y<_\lex z$ forms
an undirected hyperedge iff the longest common prefix of $x$ and $y$
is shorter than the longest common prefix of $y$ and~$z$.

>From Theorem~\ref{T-uncountable} (i.e., the existence of large cliques
or anticliques in $\omega$-automatic graphs), we derive that any
$\omega$-automatic partial order of size~$\cont$ contains an antichain
of size~$\cont$ or a copy of the real line. 

\section{Preliminaries}

\subsection{Ramsey-theory}

For a set $V$ and a natural number $k\ge1$, let $[V]^k$ denote the set
of $k$-element subsets of~$V$. A \emph{$(k,\ell)$-partition} is a pair
$G=(V,E_1,\dots,E_\ell)$ where $V$ is a set and $(E_1,\dots,E_\ell)$
is a partition of $[V]^k$ into (possibly empty) sets. For $1\le
i\le\ell$, a set $W\subseteq V$ is \emph{$E_i$-homogeneous} if
$[W]^k\subseteq E_i$; it is \emph{homogeneous} if it is
$E_i$-homogeneous for some $1\le i\le \ell$. The case $k=\ell=2$ is
special: any $(2,2)$-partition $G=(V,E_1,E_2)$ can be considered as an
(undirected loop-free) graph~$(V,E_1)$. Homogeneous sets in $G$ are
then complete or discrete induced subgraphs of~$(V,E_1)$.

Ramsey theory is concerned with the following question: Does every
$(k,\ell)$-partition $G=(V,E_1,\dots,E_\ell)$ with $|V|=\kappa$ have a
homogeneous set of size~$\lambda$ (where $\kappa$ and $\lambda$ are
cardinal numbers and $k,\ell\ge2$ are natural numbers). If this is the
case, one writes
\[
   \kappa\to(\lambda)^k_\ell
\]
(a notation due to Erd\H os and Rado \cite{ErdR56}). This allows to
formulate Ramsey's theorem concisely:

\begin{theorem}[Ramsey \cite{Ram30}]
  If $k,\ell\ge 2$, then $\alephN\to(\alephN)^k_\ell$.
\end{theorem}

In particular, every graph with~$\alephN$ nodes contains a complete or
discrete induced subgraph of the same size. If one wants to find
homogeneous sets of size $\aleph_1$, the base set has to be much
larger:

\begin{theorem}[Sierpi\'nski \cite{Sie33}]\label{T-S}
  If $k,\ell\ge2$, then $\cont\not\to(\aleph_1)^k_\ell$ and therefore
  in particular $\cont\not\to(\cont)^k_\ell$.
\end{theorem}

Erd\H os and Rado \cite{ErdR56} proved that partitions of size
properly larger than $\cont$ have homogeneous sets of size
$\aleph_1$. For more details on infinite Ramsey theory,
see~\cite[Chapter~9]{Jec02}.

\subsection{$\omega$-languages}

Let $\Gamma$ be a finite alphabet. With $\Gamma^*$ we denote the set
of all finite words over the alphabet~$\Gamma$. The set of all
nonempty finite words is $\Gamma^+$.  An \emph{$\omega$-word} over
$\Gamma$ is an infinite $\omega$-sequence $x=a_0 a_1 a_2 \cdots$ with
$a_i \in \Gamma$, we set $x[i,j)=a_ia_{i+1}\dots a_{j-1}$ for natural
numbers $i\le j$. In the same spirit, $x[i,\omega)$ denotes the
$\omega$-word $a_ia_{i+1}\dots$. The set of all $\omega$-words
over~$\Gamma$ is denoted by~$\Gamma^\omega$ and
$\Gamma^\infty=\Gamma^*\cup\Gamma^\omega$. For a set $V \subseteq
\Gamma^+$ of finite words let $V^\omega \subseteq \Gamma^\omega$ be
the set of all $\omega$-words of the form $v_0v_1 v_2 \cdots$ with
$v_i \in V$. Two infinite words~$x,y \in \Gamma^\omega$ are
\emph{ultimately equal}, briefly $x\sim_e y$, if there exists
$i\in\bN$ with $x[i,\omega)=y[i,\omega)$.  By $\le_\lex$, we denote
the lexicographic order on the set $\Sigma^\omega$ (with some,
implicitly assumed linear order on the letters from $\Sigma$) and
$\le_\pref$ the prefix order on $\Sigma^\infty$.

For $\Sigma=\{0,1\}$, the support $\supp(x)\subseteq\bN$ is the set of
positions of the letter $1$ in the word $x\in\Sigma^\omega$.

A (nondeterministic) \emph{B\"uchi-automaton}~$M$ is a tuple $M =
(Q,\Gamma, \delta, \iota, F)$ where~$Q$ is a finite set of states,
$\iota \in Q$ is the initial state, $F \subseteq Q$ is the set of
final states, and $\delta \subseteq Q \times \Gamma \times Q$ is the
transition relation. If $\Gamma=\Sigma^n$ for some alphabet $\Sigma$,
then we speak of an \emph{$n$-dimensional B\"uchi-automaton over
  $\Sigma$}. A \emph{run} of~$M$ on an $\omega$-word $x = a_0 a_1 a_2
\cdots$ is an $\omega$-word $r = p_0 p_1 p_2 \cdots$ over the set of
states~$Q$ such that $(p_i, a_i, p_{i+1}) \in \delta$ for all $i \geq
0$. The run~$r$ is \emph{successful} if $p_0=\iota$ and there exists a
final state from~$F$ that occurs infinitely often in~$r$.  The
$\omega$-language $L(M) \subseteq \Gamma^\omega$ defined by~$M$ is the
set of all $\omega$-words that admit a successful run.  An
$\omega$-language $L \subseteq \Gamma^\omega$ is \emph{regular} if
there exists a B\"uchi-automaton~$M$ with~$L(M)=L$.

Alternatively, regular $\omega$-languages can be represented
algebraically. To this end, one defines \emph{$\omega$-semigroups} to
be two-sorted algebras $S=(S_+,S_\omega;\cdot,*,\pi)$ where
$\cdot:S_+\times S_+\to S_+$ and $*:S_+\times S_\omega\to S_\omega$
are binary operations and $\pi:(S_+)^\omega\to S_\omega$ is an
$\omega$-ary operation such that the following hold:
\begin{itemize}
\item $(S_+,\cdot)$ is a semigroup,
\item $s*(t*u)=(s\cdot t)*u$,
\item $s_0\cdot\pi((s_i)_{i\ge 1})=\pi((s_i)_{i\ge0})$,
\item $\pi((s^1_i\cdot s^2_i\cdots s^{k_i}_i)_{i\ge0}) =
  \pi((t_j)_{j\ge0})$ whenever
  \[
    (t_j)_{j\ge0} =
     (s^1_0,s^2_0,\dots,s^{k_0}_0,s^1_1,\dots, s^{k_1}_1,\dots)\ .
  \] 
\end{itemize}
The $\omega$-semigroup $S$ is \emph{finite} if both, $S_+$ and
$S_\omega$ are finite. The free $\omega$-semigroup generated by
$\Gamma$ is
\[
   \Gamma^\infty=(\Gamma^+,\Gamma^\omega;\cdot,*,\pi)
\]
where $u\cdot v$ and $u*x$ are the natural operations of prefixing a
word by the finite word $u$, and $\pi((u_i)_{i\ge0})$ is the
omega-word $u_0u_1u_2\dots$. A homomorphism $h:\Gamma^\infty\to S$ of
$\omega$-semigroups maps finite words to elements of $S_+$ and
$\omega$-words to elements of~$S_\omega$ and commutes with the
operations $\cdot$, $*$, and $\pi$. The algebraic characterisation of
regular $\omega$-languages then reads as follows.

\begin{proposition}
  An $\omega$-language $L\subseteq\Gamma^\omega$ is regular if and only
  if there exists a finite $\omega$-semigroup $S$, a set $T\subseteq
  S_\omega$, and a homomorphism $\eta:\Gamma^\infty\to S$ such that
  $L=\eta^{-1}(T)$.
\end{proposition}

Hence, every B\"uchi-automaton is ``equivalent'' to a homomorphism
into some finite $\omega$-semigroup together with a distinguished
set~$T$ (and vice versa).

For $\omega$-words $x_i=a_i^0 a_i^1 a_i^2\dots \in \Gamma^\omega$,
the \emph{convolution} $x_1 \otimes x_2 \otimes \cdots \otimes
x_n\in(\Gamma^n)^\omega$ is defined by
\[
   (x_1, \dots, x_n )^\otimes
  = (a_1^0,\ldots,x_n^0)\,  (a_1^1,\ldots,a_n^1)\,
       (a_1^2,\ldots,a_n^2) \cdots\ .
\]  
An~$n$-ary relation $R \subseteq (\Gamma^\omega)^n$ is called
\emph{$\omega$-automatic} if the $\omega$-language $\{ (x_1, \dots,
x_n)^\otimes \mid (x_1, \ldots, x_n) \in R \}$ is regular.

To describe the complexity of $\omega$-languages, we will use
language-theoretic terms. Let $\LANG$ denote the class of all
languages (i.e., sets of finite words over some finite set of symbols)
and $\oLANG$ the class of all $\omega$-languages. By $\REG$ and
$\oREG$, we denote the regular languages and $\omega$-languages, resp.
An $\omega$-language is \emph{context-free} if it can be accepted by a
pushdown-automaton with B\"uchi-acceptance (on states), it is
\emph{co-context-free} if its complement is context-free. We denote by
$\oCF$ the set of context-free $\omega$-languages and by $\ocCF$ their
complements. An $\omega$-language belongs to $\LANG^*$ if it is of the
form $\bigcup_{1\le i\le n}U_iV_i^\omega$ with $U_i,V_i\in\LANG$. Then
$\oREG\subseteq\LANG^*$ and $\oCF\subseteq\LANG^*$ where the sets
$U_i$ and $V_i$ are regular and context-free, resp \cite{Sta97}. In
between these two classes, we define the class $\erCF$ of
\emph{eventually regular context-free} $\omega$-languages that
comprises all sets of the form $\bigcup_{1\le i\le n}U_iV_i^\omega$
with $U_i\in\LANG$ context-free and $V_i\in\LANG$
regular. Alternatively, eventually regular context-free
$\omega$-languages are the finite unions of $\omega$-languages of the
form $C\cdot L$ where $C$ is a context free-language and $L$ a regular
$\omega$-language. Let $\cerCF$ denote the set of complements of
eventually regular context-free $\omega$-languages.

A final, rather peculiar class of $\omega$-languages is
{\boldmath$\Lambda$}: it is the class of $\omega$-languages $L$ such
that $(\mathbb R,\le)$ embeds into $(L,\le_\lex)$ (the name derives
from the notation {\boldmath$\lambda$} for the order type of~$(\mathbb
R,\le)$).

\subsection{$\omega$-automatic $(k,\ell)$-partitions}

An \emph{$\omega$-automatic presentation of a $(k,\ell)$-partition
  $(V,E_1,\dots,E_\ell)$} is a pair $(L,h)$ consisting of a regular
$\omega$-language~$L$ and a surjection $h:L\to V$ such that
$\{(x_1,x_2,\dots,x_k)\in L^k\mid \{h(x_1),h(x_2),\dots,h(x_k)\}\in
E_i\}$ for $1\le i\le k$ and 
$R_\approx=\{(x_1,x_2)\in L^2\mid
h(x_1)=h(x_2)\}$ are $\omega$-automatic. An $\omega$-automatic
presentation is \emph{injective} if $h$ is a bijection. A
$(k,\ell)$-partition is \emph{(injectively) $\omega$-automatic} if it
has an (injective) $\omega$-automatic presentation. From
\cite{BarKR08}, it follows that an uncountable $\omega$-automatic
$(k,\ell)$-partition has $\cont$ elements.

This paper is concerned with the question whether every (injective)
$\omega$-{}auto\-matic presentation $(L,h)$ of a $(k,\ell)$-partition
admits a ``simple'' set $H\subseteq L$ such that $h(H)$ has $\lambda$
elements and is homogeneous. More precisely, let $\mathcal C$ be a
class of $\omega$-languages, $k,\ell\ge2$ natural numbers, and
$\kappa$ and $\lambda$ cardinal numbers. Then we write
\[
   (\kappa,\omega\mathsf{A})\ramsey(\lambda,\mathcal C)^k_\ell
\]
if the following partition property holds: for every
$\omega$-automatic presentation~$(L,h)$ of a $(k,\ell)$-partition~$G$
of size $\kappa$, there exists $H\subseteq L$ in $\mathcal C$ such
that $h(H)$ is homogeneous in $G$ and of size $\lambda$. 
\[
   (\kappa,\omega\mathsf{iA})\ramsey(\lambda,\mathcal C)^k_\ell
\]
is to be understood similarly where we only consider injective
$\omega$-automatic presentations.

\begin{remark}
  Let $G=(V,E_1,\dots,E_\ell)$ be some $(k,\ell)$-partition with
  $\omega$-automatic presentation $(L,h)$. Then the partition property
  above requires that there is a ``large'' homogeneous set $X\subseteq
  V$ and an $\omega$-language $H\in\mathcal C$ such that $h(H)=X$, in
  particular, every element of $X$ has at least one representative
  in~$H$. Alternatively, one could require that $h^{-1}(X)\subseteq L$
  is an $\omega$-language from~$\mathcal C$. In this paper, we only
  encounter classes $\mathcal C$ of $\omega$-languages such that the
  following closure property holds: if $H\in\mathcal C$ and $R$ is an
  $\omega$-automatic relation, then also $R(H)=\{y \mid \exists x\in
  H:(x,y)\in R\}\in\mathcal C$. Since
  $h^{-1}h(H)=R_\approx(H)$, all our results also hold for this
  alternative requirement $h^{-1}(X)\in\mathcal C$.
\end{remark}

This paper shows
\begin{enumerate}[(1)]
\addtocounter{enumi}{-1}
\item if $k,\ell\ge2$, then
  $(\alephN,\omega\mathsf{A})\ramsey(\alephN,\omega\mathrm{REG})^k_\ell$,
  but
  $(\cont,\omega\mathsf{A})\not\ramsey(\alephN,\omega\mathrm{REG})^k_\ell$,
  see Theorem~\ref{T-countable}.
\item if $\ell\ge2$, then
  $(\cont,\omega\mathsf{A})\ramsey(\cont,\cerCF)^2_\ell$, see
  Theorem~\ref{T-uncountable}.
\item if $k\ge3$, $\ell\ge2$, and $\lambda>\alephN$, then
  $(\cont,\omega\mathsf{iA})\not\ramsey(\lambda,\oLANG)^k_\ell$, see
  Theorem~\ref{T-S2}.
\item if $k,\ell\ge2$ and $\lambda>\alephN$, then
  $(\cont,\omega\mathsf{iA})\not\ramsey(\lambda,\oCF)^k_\ell$, see
  Theorem~\ref{T-complexity}.
\end{enumerate}

Here, the first part of (0) is a strengthening of Ramsey's theorem
since the infinite homogeneous set is regular. The second part might
look surprising since larger $(k,\ell)$-partitions should have larger
homogeneous sets -- but not necessarily regular ones! In contrast to
Sierpi\'nski's result, (1) shows that $\omega$-automatic
$(2,\ell)$-partitions have a larger degree of homogeneity than
arbitrary $(2,\ell)$-partitions. Even more, the complexity of the
homogeneous set can be bound in language-theoretic terms (there is
always a homogeneous set that is the complement of an eventually
regular context-free $\omega$-language).  Statement (2) is an analogue
of Sierpi\'nski's Theorem~\ref{T-S} showing that (injective)
$\omega$-automatic $(k,\ell)$-partitions are as in-homogeneous as
arbitrary $(k,\ell)$-partitions provided $k\ge3$. The complexity bound
from (1) is shown to be optimal by~(3) proving that one cannot always
find context-free homogeneous sets.  Hence, despite the existence of
large homogeneous sets for $k=2$, for some $\omega$-automatic
presentations, they are bound to have a certain (low) level of
complexity that is higher than the regular $\omega$-languages.

\section{Countably infinite homogeneous sets}

Let $k,\ell\ge2$ be arbitrary. Then, from Ramsey's theorem, we obtain
immediately $(\alephN,\omega\mathsf{A})\ramsey(\alephN,\oLANG)^k_\ell$
and $(\cont,\omega\mathsf{A})\ramsey(\alephN,\oLANG)^k_\ell$, i.e.,
all infinite $\omega$-automatic $(k,\ell)$-partitions have homogeneous
sets of size~$\alephN$. In this section, we ask whether such
homogeneous sets can always be chosen regular:

\begin{theorem}\label{T-countable}
  Let $k,\ell\ge2$. Then
  \begin{enumerate}[(a)]
  \item $(\alephN,\omega\mathsf{A})\ramsey(\alephN,\omega\mathrm{REG})^k_\ell$.
  \item $(\cont,\omega\mathsf{iA})\ramsey(\alephN,\omega\mathrm{REG})^k_\ell$.
  \item $(\cont,\omega\mathsf{A})\not\ramsey(\alephN,\LANG^*)^k_\ell$,
    and therefore in particular
    $(\cont,\omega\mathsf{A})\not\ramsey(\alephN,\oCF)^k_\ell$ and
    $(\cont,\omega\mathsf{A})\not\ramsey(\alephN,\oREG)^k_\ell$.
  \end{enumerate}
\end{theorem}

\begin{proof}
  Let $(L,h)$ be an $\omega$-automatic presentation of some
  $(k,\ell)$-partition~$G=(V,E_1,\dots,E_\ell)$ with $|V|=\alephN$. By
  \cite{BarKR08}, there exists $L'\subseteq L$ regular such that
  $(L',h)$ is an injective $\omega$-automatic presentation
  of~$G$. From a B\"uchi-automaton for $L'$, one can compute a finite
  automaton accepting some language $K$ such that $(K,h')$ is an
  injective automatic presentation of~$G$~\cite{Blu99}. Hence, by
  \cite{Rub08}, there exists a regular set $H'\subseteq K$ such that
  $h'(H')$ is homogeneous in~$G$ and countably infinite. From this
  set, one obtains a regular $\omega$-language $H\subseteq L'\subseteq
  L$ with $h(H)=h'(H')$, i.e., $h(H)$ is a homogeneous set of size
  $\alephN$. This proves~(a).

  To prove (b), let $(L,h)$ be an injective $\omega$-automatic
  presentation of some $(k,\ell)$-partition $G=(V,E_1,\dots,E_\ell)$
  of size~$\cont$. Then there exists a regular $\omega$-language
  $L'\subseteq L$ with $|L'|=\alephN$. Consider the sub-partition
  $G'=(h(L'),E_1',\dots,E_\ell')$ with $E_i'=E_i\cap[h(L')]^k$. This
  $(k,\ell)$-partition has as $\omega$-automatic presentation the
  pair~$(L',h)$.  Then, by (a), there exists $L''\subseteq
  L'$ regular and infinite such that $h(L'')$ is homogeneous in $G'$
  and therefore in~$G$. Since $h$ is injective, this implies
  $|h(L')|=|L'|=\alephN$.

  Finally, we show (c) by a counterexample. Let $L=\{0,1\}^\omega$,
  $V=L/\mathord{\sim_e}$, and $h:L\to V$ the canonical
  mapping. Furthermore, set $E_1=[L]^k$. Then
  $G=(V,E_1,\emptyset,\dots,\emptyset)$ is a $(k,\ell)$-partition with
  $\omega$-automatic presentation $(L,h)$.

  Now let $H=\bigcup_{1\le i\le n}U_iV_i^\omega\subseteq L$ for some
  non-empty languages $U_i,V_i\subseteq\{0,1\}^+$ such that $h(H)$ is
  homogeneous and infinite.

  If $|V_i^\omega|=1$, then $U_iV_i^\omega/\mathord{\sim_e}$ is
  finite. Since $h(H)$ is infinite, there exists $1\le i\le n$ with
  $|V_i^\omega|>1$ implying the existence of words $v,w\in V_i^+$ such
  that $|v|=|w|$ and $v\neq w$. For $u\in U_i$, the set
  $u\{v,w\}^\omega\subseteq H$ has $\cont$ equivalence classes
  wrt.~$\sim_e$. Hence $|h(H)|=\cont$.
\end{proof}

\section{Uncountable homogeneous sets}

\subsection{A Ramsey theorem for $\omega$-automatic $(2,\ell)$-partitions}

The main result of this section is the following theorem that follows
immediately from Prop.~\ref{P-ramsey2} and Lemma~\ref{L-H}.

\begin{theorem}\label{T-uncountable}
  For all $\ell\ge2$, we have
  $(\cont,\omega\mathsf{A})\ramsey(\cont,\cerCF\cap{\text{\boldmath$\Lambda$}})^2_\ell$.
\end{theorem}

\subsubsection{The proof}

The proof of this theorem will construct a language from $\cerCF$ that
describes a homogeneous set. This language is closely related to the
following language
\[
   N=1\{0,1\}^\omega\cap
      \bigcap_{n\ge0} \{0,1\}^n
                         (0\{0,1\}^n00\cup 10^n\{01,10\})
                            \{0,1\}^\omega\ ,
\]
i.e., an $\omega$-word~$x$ belongs to $N$ iff it starts with $1$ and, for
every $n\ge0$, we have $x[n,2n+3)\in 0\{0,1\}^*00\cup 10^*01\cup
10^*10$. We first list some useful properties of this language~$N$:

\begin{lemma}\label{L-almost-disjoint}
  The $\omega$-language $N$ is contained in $(1^+0^+)^\omega$, belongs
  to $\cerCF\cap\text{\boldmath$\Lambda$}$, and $\supp(x)\cap\supp(y)$
  is finite for any $x,y\in N$ distinct.
\end{lemma}

\begin{proof}
  Let $b_i\in\{0,1\}$ for all $i\ge0$ and suppose the word
  $x=b_0b_1\dots$ belongs to~$N$. Then $b_0=1$, hence the word $x$
  contains at least one occurrence of $1$. Note that, whenever
  $b_n=1$, then $\{b_{2n+1},b_{2n+2}\}=\{0,1\}$, hence $x$ contains
  infinitely many occurrences of $1$ and therefore infinitely many
  occurrences of~$0$, i.e., $N\subseteq(1^+0^+)^\omega$.

  Note that the complement of $N$ equals
  \begin{align*}
    &\ 0\{0,1\}^\omega\cup
    \bigcup_{n\ge0}\Big(\{0,1\}^n(0\{0,1\}^n\{01,10,11\}\cup
    1\{0,1\}^n\{00,11\})\{0,1\}^\omega\Big)\\
    =& \left[0\cup
    \bigcup_{n\ge0}\{0,1\}^n(0\{0,1\}^n\{01,10,11\}\cup
    1\{0,1\}^n\{00,11\})\right]\{0,1\}^\omega\ .
  \end{align*}
  Since the expression in square brackets denotes a context-free
  language, $\{0,1\}^\omega\setminus N$ is an eventually regular
  context-free $\omega$-language. 

  Note that a word $10^{n_0}10^{n_1}10^{n_2}\dots$ belongs to $N$ iff,
  for all $k\ge0$, we have $0\le n_k-|10^{n_0}10^{n_1}\dots
  10^{n_{k-1}}|\le 1$. Hence, when building a word from $N$, we have
  two choices for any $n_k$, say $n_k^0$ and $n_k^1$ with
  $n_k^0<n_k^1$. But then $a_0a_1a_2\dots\mapsto
  10^{n_0^{a_0}}10^{n_1^{a_1}}10^{n_2^{a_2}}\dots$ defines an order
  embedding
  $(\{0,1\}^\omega,\le_\lex)\hookrightarrow(N,\le_\lex)$. Since
  $(\mathbb R,\le)\hookrightarrow(\{0,1\}^\omega,\le_\lex)$, we get
  $N\in\text{\boldmath$\Lambda$}$.

  Now let $x,y\in N$ with $\supp(x)\cap\supp(y)$ infinite. Then there
  are arbitrarily long finite words $u$ and $v$ of equal length such
  that $u1$ and $v1$ are prefixes of $x$ and $y$, resp. Since $u1$ is
  a prefix of $x\in N$, it is of the form $u1=u'10^{|u'|}1$ (if $|u|$
  is even) or $u1=u'10^{|u'|}01$ (if $|u|$ is odd) and analogously
  for~$v$. Inductively, one obtains $u'=v'$ and therefore~$u=v$. Since
  $u$ and $v$ are arbitrarily long, we showed $x=y$. 
\end{proof}

\begin{lemma}\label{L-listing}
  Let $\sim$ and $\approx$ be two equivalence relations on some
  set~$L$ such that any equivalence class $[x]_\sim$ of $\sim$ is
  countable and $\approx$ has $\cont$ equivalence classes. Then there
  are elements $(x_\alpha)_{\alpha<\cont}$ of $L$ such that
  $[x_\alpha]_{\sim_e}\cap[x_\beta]_\approx=\emptyset$ for all
  $\alpha<\beta$.
\end{lemma}

\begin{proof}
  We construct the sequence $(x_\alpha)_{\alpha<\cont}$ by ordinal
  induction. So assume we have elements $(x_\alpha)_{\alpha<\kappa}$
  for some ordinal $\kappa<\cont$ with
  $[x_\alpha]_{\sim}\cap[x_\beta]_\approx=\emptyset$ for all
  $\alpha<\beta<\kappa$.

  Suppose
  $\bigcup_{\alpha<\kappa}[x_\alpha]_{\sim}\cap[x]_\approx\neq\emptyset$
  for all $x\in L$. For $x,y\in L$ with $x\not\approx y$, we have
  $(\bigcup_{\alpha<\kappa}[x_\alpha]_{\sim}\cap[x]_\approx)\cap(\bigcup_{\alpha<\kappa}[x_\alpha]_{\sim}\cap[y]_\approx)\subseteq[x]_\approx\cap[y]_\approx=\emptyset$. Since
  $\bigcup_{\alpha<\kappa}[x_\alpha]_{\sim}$ has
  $\kappa\cdot\alephN\le\max(\kappa,\aleph_0)<\cont$ elements, we
  obtain $|L|<\cont$, contradicting $|L|\ge|L/{\approx}|=\cont$.
  Hence there exists an element $x_\kappa\in L$ with
  $[x_\alpha]_{\sim}\cap[x_\kappa]_\approx=\emptyset$ for all
  $\alpha<\kappa$.
\end{proof}

\begin{definition}\label{D-H}
  Let $u$, $v$, and $w$ be nonempty words with
  $|v|=|w|$ and $v\neq w$. Define an
  $\omega$-semigroup homomorphism $h:\{0,1\}^\infty\to\Sigma^\infty$
  by $h(0)=v$ and $h(1)=w$ and set
  \[
     H_{u,v,w}=u\cdot h(N)
  \]
  where $N$ is the set from Lemma~\ref{L-almost-disjoint}.
\end{definition}

\begin{lemma}\label{L-H}
  Let $u$, $v$, and $w$ be as in the previous
  definition. Then
  $H_{u,v,w}\in\cerCF\cap\text{\boldmath$\Lambda$}$.
\end{lemma}

\begin{proof}
  Assume $v<_\lex w$. Then the mapping
  $\chi:\{0,1\}^\omega\to\Sigma^\omega:x\mapsto u h(x)$ (where
  $h$ is the homomorphism from the above definition) embeds
  $(N,\le_\lex)$ (and hence $(\mathbb R,\le)$)
  into~\mbox{$(H_{u,v,w},\le_\lex)$}.  If
  $w<_\lex v$, then
$(\mathbb R,\le)\cong(\mathbb R,\ge)
      \hookrightarrow (N,\ge_\lex)\hookrightarrow 
            (H_{\alpha,\beta,\gamma},\le_\lex)$.
  This proves that $H_{u,v,w}$ belongs to
  {\boldmath$\Lambda$}.

  Since $v\neq w$, the mapping $\chi$ is injective. Hence
  \[
    \Sigma^\omega\setminus H_{\alpha,\beta,\gamma}
      = \Sigma^\omega\setminus\chi(N)
      =\Sigma^\omega\setminus\chi(\{0,1\}^\omega)\cup
          \chi(\{0,1\}^\omega\setminus N)\ .
  \]
  Since $\chi$ can be realized by a generalized sequential
  machine with B\"uchi-acceptance, $\chi(\{0,1\}^\omega)$
  is regular and $\chi(\{0,1\}^\omega\setminus N)$ (as the image of an
  eventually regular context-free $\omega$-language) is eventually
  regular context-free. Hence $\Sigma^\omega\setminus
  H_{u,v,w}$ is eventually regular context-free.
\end{proof}

\begin{proposition}\label{P-ramsey}
  Let $G=(L,E_0,E_1,\dots,E_\ell)$ be some $(2,1+\ell)$-partition with
  injective $\omega$-automatic presentation $(L,\mathrm{id})$ such
  that $\{(x,y)\mid \{x,y\}\in E_0\}\cup\{(x,x)\mid x\in L\}$ is an
  equivalence relation on~$L$ (denoted $\approx$) with $\cont$
  equivalence classes. Then there exist nonempty words $u$,
  $v$, and $w$ with $v$ and $w$ distinct, but of the
  same length, such that $H_{u,v,w}$ is $i$-homogeneous
  for some $1\le i\le\ell$.
\end{proposition}

\begin{proof}
  There are finite $\omega$-semigroups $S$ and $T$ and homomorphisms
  $\gamma:\Sigma^\infty\to S$ and
  $\delta:(\Sigma\times\Sigma)^\infty\to T$ such that
  \begin{enumerate}[(a)]
  \item $x\in L$, $y\in\Sigma^\omega$, and $\gamma(x)=\gamma(y)$ imply
    $y\in L$ and
  \item $x,x',y,y'\in L$, $\{h(x),h(x')\}\in E_i$, and
    $\delta(x,x')=\delta(y,y')$ imply $\{h(y),h(y')\}\in E_i$ (for all
    $0\le i\le\ell$).
  \end{enumerate}
  By Lemma~\ref{L-listing}, there are words
  $(x_\alpha)_{\alpha<\cont}$ in $L$ such that
  $[x_\alpha]_{\sim_e}\cap[x_\beta]_\approx=\emptyset$ for all
  $\alpha<\beta$.

  In the following, we only need the words $x_0,x_1,\dots,x_C$ with
  $C=|S|\cdot|T|$. Then \cite[Sections 3.1-3.3]{BarKR08}\footnote{The
    authors of \cite{BarKR08} require
    $[x_i]_{\sim_e}\cap[x_j]_\approx=\emptyset$ for all $0\le i,j\le
    C$ distinct, but they use it only for $i<j$. Hence we can apply
    their result here.} first constructs two $\omega$-words $y_1$ and
  $y_2$ and an infinite sequence $1\le g_1<g_2<\dots$ of natural
  numbers such that in particular $y_1[g_1,g_2)<_\lex
  y_2[g_1,g_2)$. Set $u=y_2[0,g_1)$, $v=y_1[g_1,g_2)$, and
  $w=y_2[g_1,g_2)$. In the following, let
  $h:\{0,1\}^\infty\to\Sigma^\infty$ be the homomorphism from
  Def.~\ref{D-H} and set $\chi(x)= u h(x)$ for $x\in\{0,1\}^*$.  As
  in \cite{BarKR08}, one can then show that all the words from $H_{u,v,w}$
  belong to the $\omega$-language~$L$. In the following, set
  $x_{\circ\bullet}=\chi((01)^\omega)$ and
  $x_{\bullet\circ}=\chi((10)^\omega)$. Then obvious alterations in
  the proofs by B{\'a}r{\'a}ny et al.\ show:
  \begin{enumerate}[(1)]
  \item \cite[Lemma 3.4]{BarKR08}\footnote{The authors of
      \cite{BarKR08} only require one of the two differences to be
      infinite, but the proof uses that they both are infinite.}  If
    $x,y\in\{0,1\}^\omega$ with $\supp(x)\setminus\supp(y)$ and
    $\supp(y)\setminus\supp(x)$ infinite, then
    \[
       \{\delta(\chi(x),\chi(y)),\delta(\chi(y),\chi(x))\}
        =\{\delta(x_{\bullet\circ},x_{\circ\bullet}),
       \delta(x_{\circ\bullet},x_{\bullet\circ})\}\ .
    \]
  \item \cite[Lemma 3.5]{BarKR08} $x_{\bullet\circ}\not\approx
    x_{\circ\bullet}$.
  \end{enumerate}
  There exists $0\le i\le \ell$ with
  $\{x_{\bullet\circ},x_{\circ\bullet}\}\in E_i$. Then (2) implies
  $i>0$.

  Let $x,y\in N$ be distinct. Then $\supp(x)\cap\supp(y)$ is finite by
  Lemma~\ref{L-almost-disjoint}. Since, on the other hand, $\supp(x)$
  and $\supp(y)$ are both infinite, the two differences
  $\supp(x)\setminus\supp(y)$ and $\supp(y)\setminus\supp(x)$ are
  infinite. Hence we obtain $\delta(\chi(x),\chi(y))\in
  \{\delta(x_{\bullet\circ},x_{\circ\bullet}),
  \delta(x_{\circ\bullet},x_{\bullet\circ})\}$ from (1).  Hence (b)
  implies $\{\chi(x),\chi(y)\}\in E_i$, i.e., $H_{u,v,w}$ is $E_i$-homogeneous. 

  Since $H_{u,v,w}\in\cerCF\cap\text{\boldmath$\Lambda$}$ by Lemma~\ref{L-H}, the result follows.
\end{proof}

\begin{proposition}\label{P-ramsey2}
  Let $G=(V,E_1',\dots,E_\ell')$ be some $(2,\ell)$-partition with
  automatic presentation $(L,h)$. Then there exist $u,v,w\in\Sigma^+$
  with $v$ and $w$ distinct of equal length such that $h(H_{u,v,w})$
  is homogeneous and of size $\cont$.
\end{proposition}

\begin{proof}
  To apply Prop.~\ref{P-ramsey}, consider the following
  $(2,1+\ell)$-partition $G=(L,E_0,\dots,E_\ell)$:
  \begin{itemize}
  \item The underlying set is the $\omega$-language~$L$,
  \item $E_0$ comprises all sets $\{x,y\}$ with $h(x)=h(y)$ and $x\neq y$, and
  \item $E_i$ (for $1\le i\le\ell$) comprises all sets $\{x,y\}$ with
    $\{h(x),h(y)\}\in E_i'$.
  \end{itemize}
  Then $(L,\mathrm{id})$ is an injective $\omega$-automatic
  presentation of the $(2,1+\ell)$-partition~$G$. By
  Prop.~\ref{P-ramsey}, there exists $1\le i\le\ell$ and words $u$,
  $v$ and $w$ such that $H_{u,v,w}$ is $i$-homogeneous in~$G$. Since
  $(E_0,\dots,E_\ell)$ is a partition of $[L]^2$, we have
  $\{x,y\}\notin E_0$ (and therefore $h(x)\neq h(y)$) for all $x,y\in
  H_{u,v,w}$ distinct. Hence
  $h$ is injective on $H_{u,v,w}$. Furthermore
  $[H_{u,v,w}]^2\subseteq E_i$ implies $[h(H_{u,v,w})]^2\subseteq
  E_i'$. Hence $h(H_{u,v,w})$ is an $i$-homogeneous set in $G'$ of
  size $\cont$.
\end{proof}

This finishes the proof of Theorem~\ref{T-uncountable}. 

\subsubsection{Effectiveness}

Note that the proof above is non-constructive at several points:
Lemma \ref{L-listing} is not constructive and the proof proper uses
Ramsey's theorem \cite[page 390]{BarKR08} and makes a Ramseyan
factorisation coarser \cite[begin of section 3.2]{BarKR08}. We now
show that nevertheless
the words $u$, $v$, and $w$ can be computed. By Prop.~\ref{P-ramsey2},
it suffices to decide for a given triple $(u,v,w)$ whether
$h(H_{u,v,w})$ is $i$-homogeneous for some fixed $1\le i\le \ell$.

To be more precise, let $(V,E_1,\dots,E_\ell)$ be some
$(2,\ell)$-partition with $\omega$-automatic presentation $(L,h)$. Furthermore, let
$u,v,w\in\Sigma^+$ with $v\neq w$ of the same length and write $H$ for
$H_{u,v,w}$. We have to decide whether $H\subseteq L$ and $H\otimes
H\subseteq L_i\cup L_=$. Note that $H\subseteq L$ iff
$L\cap\Sigma^\omega\setminus H=\emptyset$. But $\Sigma^\omega\setminus
H$ is context-free, so the intersection is context-free. Hence the
emptiness of the intersection can be decided. 

Towards a decision of the second requirement, note that 
\begin{align*}
  (\Sigma\times\Sigma)^\omega\setminus (H\otimes H)
    & = (\Sigma^\omega\setminus H\otimes\Sigma^\omega)
       \cup
        (\Sigma^\omega\cup\Sigma^\omega\setminus H)
\end{align*}
is the union of two context-free $\omega$-languages and therefore
context-free itself. Since $L_i\cup L_=$ is regular, the intersection
$(L_i\cup L_=)\cap(\Sigma\times\Sigma)^\omega\setminus(H\otimes H)$ is
context-free implying that its emptiness is decidable. But this
emptiness is equivalent to $H\otimes H\subseteq L_1\cup L_=$.

\subsubsection{$\omega$-automatic partial orders}

>From Theorem~\ref{T-uncountable}, we now derive a necessary condition
for a partial order of size~$\cont$ to be $\omega$-automatic. A
partial order $(V,\sqsubseteq)$ is $\omega$-automatic iff there exists
a regular $\omega$-language~$L$ and a surjection $h:L\to V$ such that
the relations $R_==\{(x,y)\in L^2\mid h(x)=h(y)\}$ and
$R_\sqsubseteq=\{(x,y)\in L^2\mid h(x)\sqsubseteq h(y)\}$ are
$\omega$-automatic.

\begin{corollary}[\cite{BarKR08}\footnote{As pointed out by two
    referees, the paragraph before Sect.~4.1 in \cite{BarKR08} already
    hints at this result, although in a rather implicit
    way.}\addtocounter{footnote}{-1}]
  If $(V,\sqsubseteq)$ is an $\omega$-automatic partial order with
  $|V|\ge\aleph_1$, then $(\mathbb R,\le)$ or an antichain of
  size~$\cont$ embeds into $(V,\sqsubseteq)$.
\end{corollary}

\begin{proof}
  Let $(V,\sqsubseteq)$ be a partial order, $L\subseteq\Sigma^\omega$
  a regular $\omega$-language and $h:L\to V$ a surjection such that
  $R_=$ and $R_\sqsubseteq$ are $\omega$-automatic. Define an injective
  $\omega$-automatic $(2,4)$-partition $G=(L,E_0,E_1,E_2,E_3)$:
  \begin{itemize}
   \item $E_0$ comprises all pairs $\{x,y\}\in[L]^2$ with $h(x)=h(y)$,
   \item $E_1$ comprises all pairs $\{x,y\}\in[L]^2$ with
     $h(x)\sqsubset h(y)$ and $x<_\lex y$, 
   \item $E_2$ comprises all pairs $\{x,y\}\in[L]^2$ with
     $h(x)\sqsupset h(y)$ and $x<_\lex y$, and
   \item $E_3=[L]^2\setminus(E_0\cup E_1\cup E_2)$ comprises all pairs
     $\{x,y\}\in[L]^2$ such that $h(x)$ and $h(y)$ are incomparable.
  \end{itemize}
  From $|L|\ge|V|>\alephN$, we obtain $|L|=\cont$. Hence, by
  Prop.~\ref{P-ramsey}, there exists $H\subseteq L$ $1$-, $2$- or
  $3$-homogeneous with $(\mathbb
  R,\le)\hookrightarrow(H,\le_\lex)$. Since $[H]^2\subseteq E_1\cup
  E_2\cup E_3$ and since $G$ is a partition of $L$, the mapping $h$
  acts injectively on~$H$. If $[H]^2\subseteq E_1$ (the case
  $[H]^2\subseteq E_2$ is symmetrical) then $(\mathbb
  R,\le)\hookrightarrow(H,\le_\lex)\cong(h(H),\sqsubseteq)$. If
  $[H]^2\subseteq E_3$, then $h(H)$ is an antichain of size~$\cont$.
\end{proof}

A linear order $(L,\sqsubseteq)$ is \emph{scattered} if $(\mathbb
Q,\le)$ cannot be embedded into $(L,\sqsubseteq)$. Automatic partial
orders are defined similarly to $\omega$-automatic partial orders with
the help of finite automata instead of B\"uchi-automata.

\begin{corollary}[\cite{BarKR08}\footnotemark]
  Any scattered $\omega$-automatic linear order $(V,\sqsubseteq)$ is
  countable. Hence,
  \begin{itemize}
  \item a scattered linear order is $\omega$-automatic if and only if
    it is automatic, and
  \item an ordinal~$\alpha$ is $\omega$-automatic if and only if
    $\alpha<\omega^\omega$.
  \end{itemize}
\end{corollary}

\begin{proof}
  If $(V,\sqsubseteq)$ is not countable, then it embeds $(\mathbb
  R,\le)$ by the previous corollary and therefore in particular
  $(\mathbb Q,\le)$.  The remaining two claims follow immediately from
  \cite{BarKR08} (``countable $\omega$-automatic structures are
  automatic'') and \cite{Del04} (``an ordinal is automatic iff it is
  properly smaller than $\omega^\omega$''), resp.
\end{proof}

Contrast Theorem~\ref{T-uncountable} with Theorem~\ref{T-S}: any
uncountable $\omega$-automatic $(k,\ell)$-partition contains an
uncountable homogeneous set of size $\cont$. But we were able to prove
this for $k=2$, only. One would also wish the homogeneous set to be
regular and not just from $\cerCF$. We now prove that these two
shortcomings are unavoidable: Theorem~\ref{T-uncountable} does not
hold for $k=3$ nor is there always an $\omega$-regular homogeneous
set. These negative results hold even for injective presentations.

\subsection{A Sierpi\'nski theorem for $\omega$-automatic
  $(k,\ell)$-partitions with $k\ge3$}

We first concentrate on the question whether some form of
Theorem~\ref{T-uncountable} holds for $k\ge3$. The following lemma
gives the central counterexample for $k=3$ and $\ell=2$, the below
theorem then derives the general result.

\begin{lemma}\label{L-S2}
  $(\cont,\omega\mathsf{iA})\not\ramsey(\aleph_1,\oLANG)^3_2$.
\end{lemma}

\begin{proof}
  Let $\Sigma=\{0,1\}$, $V=L=\{0,1\}^\omega$. Furthermore, for
  $H\subseteq L$, we write $\bigwedge H\in\Sigma^\infty$ for the
  longest common prefix of all $\omega$-words in~$H$,
  $\bigwedge\{x,y\}$ is also written $x\wedge y$. Then let $E_1$
  consist of all 3-sets $\{x,y,z\}\in [L]^3$ with $x<_\lex y<_\lex z$
  and $x\wedge y<_\pref y\wedge z$; $E_2$ is the complement
  of~$E_1$. This finishes the construction of the $(3,2)$-partition
  $(V,E_1,E_2)$ of size $\cont$ with injective $\omega$-automatic
  presentation~$(L,\mathrm{id})$.

  Note that $1^*0^\omega$ is a countable $E_1$-homogeneous set and
  that $0^*1^\omega$ is a countable $E_2$-homogeneous set. But there
  is no uncountable homogeneous set: First suppose $H\subseteq L$ is
  infinite and $x\wedge y<_\pref y\wedge z$ for all $x<_\lex y<_\lex
  z$ from~$H$. Let $u\in\Sigma^*$ such that $H\cap u0\Sigma^\omega$
  and $H\cap u1\Sigma^\omega$ are both nonempty and let $x,y\in H\cap
  u0\Sigma^\omega$ with $x\le_\lex y$ and $z\in H\cap
  u1\Sigma^\omega$. Then $x\wedge y>_\pref u=y\wedge z$ and therefore
  $x=y$ (for otherwise, we would have $x<_\lex y<_\lex z$ in $H$ with
  $x\wedge y>_\pref y\wedge z$). Hence we showed $|H\cap
  u0\Sigma^\omega|=1$. Let $u_0=\bigwedge H$ and $H_1=H\cap
  u_01\Sigma^\omega$. Since $H\cap u_00\Sigma^\omega$ is finite, the
  set $H_1$ is infinite. We proceed by induction: $u_n=\bigwedge H_n$
  and $H_{n+1}=H_n\cap u_n1\Sigma^\omega$ satisfying $|H_n\cap
  u_n0\Sigma^\omega|=1$. Then $u_0<_\pref u_01\le_\pref u_1<_\pref
  u_11\le_\pref u_2\cdots$ with
  \[
    H=\bigcup_{n\ge0} (H\cap u_n0\Sigma^\omega) \cup 
       \bigcap_{n\ge0} (H\cap u_n1\Sigma^\omega)\ .
  \]
  Then any of the sets $H\cap u_n0\Sigma^\omega=H_n\cap
  u_n0\Sigma^\omega$ and $\bigcap (H\cap u_n1\Sigma^\omega)$ is a
  singleton, proving that $H$ is countable. Thus, there cannot be an
  uncountable $E_1$-homogeneous set.

  So let $H\subseteq L$ be infinite with $x\wedge y\ge_\pref y\wedge
  z$ for all $x<_\lex y<_\lex z$. Since we have only two letters, we
  get $x\wedge y>_\pref y\wedge z$ for all $x<_\lex y<_\lex z$ which
  allows to argue symmetrically to the above. Thus, indeed, there is
  no uncountable homogeneous set in~$L$.
\end{proof}

\begin{theorem}\label{T-S2}
  For all $k\ge3$, $\ell\ge2$, and $\lambda>\alephN$, we have
  $(\cont,\omega\mathsf{iA})\not\ramsey(\lambda,\oLANG)^k_\ell$.
\end{theorem}

\begin{proof}
  Let $G$ be the $(3,2)$-partition from Lemma~\ref{L-S2} that does not
  have homogeneous sets of size~$\lambda$ and let $(L,\mathrm{id})$
  be an injective $\omega$-automatic presentation of~$G=(V,E_1,E_2)$
  (in particular, $V=L$).

  For a set $X\in[L]^k$, let $X_1<_\lex X_2<_\lex X_3$ be the three
  lexicographically least elements of~$X$. Then set
  $G'=(V,E_1',E_2',\dots,E_\ell')$ with
  \begin{align*}
    E_1'&=\{X\in [V]^k\mid \{X_1,X_2,X_3\}\in E_1\},\\
    E_2'&=\{X\in [V]^k\mid \{X_1,X_2,X_3\}\in E_2\},\text{ and }\\
    E_i'&=\emptyset\text{ for }3\le i\le\ell\ .
  \end{align*}
  Then $(L,\mathrm{id})$ is an injective $\omega$-automatic
  presentation of~$G'$. Now suppose $H'\subseteq L$ is homogeneous
  in~$G'$ and of size~$\lambda$. Then there exists $H\subseteq H'$ of
  size $\lambda$ such that for any words $x_1<_\lex x_2<_\lex x_3$
  from $H$, there exists $X\subseteq H'$ with $X_i=x_i$ for $1\le i\le
  3$ (if necessary, throw away some lexicographically largest elements
  of $H'$). Hence $H$ is homogeneous in $G$, contradicting
  Lemma~\ref{L-S2}.
\end{proof}

\subsection{Complexity of homogeneous sets in $\omega$-automatic
  $(2,\ell)$-partitions}

Having shown that $k=2$ is a central assumption in
Theorem~\ref{T-uncountable}, we now turn to the question whether
homogeneous sets of lower complexity can be found. 

\paragraph{Construction}
Let $V=L$ denote the regular $\omega$-language
$(1^+0^+)^\omega$. Furthermore, $E_1\subseteq[L]^2$ comprises all
2-sets $\{x,y\}\subseteq L$ such that $\supp(x)\cap\supp(y)$ is finite
or $x\sim_e y$. The set $E_2$ is the complement of $E_1$ in
$[L]^2$. This completes the construction of the $(2,2)$-partition
$G=(L,E_1,E_2)$.  Note that $(L,\mathrm{id}_L)$ is an injective
$\omega$-automatic presentation of~$G$. \medskip

By Theorem~\ref{T-uncountable}, $G$ has an $E_1$- or an
$E_2$-homogeneous set of size $\cont$.  We convince ourselves that $G$
has large homogeneous sets of both types. By
Lemma~\ref{L-almost-disjoint}, there is an $\omega$-language
$N\subseteq (1^+0^+)^\omega$ of size~$\cont$ such that the supports of
any two words from $N$ have finite intersection. Hence $[N]^2\subseteq
E_1$ and $N$ has size~$\cont$. But there is also an $E_2$-homogeneous
set $L_2$ of size $\cont$: Note that the words from $N$ are mutually
non-$\sim_e$-equivalent and let $L_2$ denote the set of all words
$1a_1 1a_21a_3\dots$ for $a_1a_2a_3\dots\in N$. Then for any $x,y\in
L_2$ distinct, we have $2\bN\subseteq\supp(x)\cap\supp(y)$ and
$x\not\sim_e y$, i.e., $\{x,y\}\in E_2$.

\begin{lemma}
  Let $H\in\LANG^*$ have size $\lambda>\alephN$. Then $H$ is not
  homogeneous in~$G$.
\end{lemma}

\begin{proof}
  By definition of $\LANG^*$, there are languages $U_i,V_i\in\LANG$
  with $H=\bigcup_{1\le i\le n}U_iV_i^\omega$.

  Since $H$ is infinite, there are $1\le i\le n$ and $x,y\in
  U_iV_i^\omega$ distinct with $x\sim_e y$ and therefore $\{x,y\}\in
  E_1$.

  Since $|H|>\alephN$, there is $1\le i\le n$ with
  $|U_iV_i^\omega|>\alephN$; we set $U=U_i$ and $V=V_i$. From
  $|U|\le\alephN$, we obtain $|V^\omega|>\alephN$. Hence there are
  $v_1,v_2\in V^+$ distinct with $|v_1|=|v_2|$. Since $uv_1^\omega\in
  H$ and each element of $H$ contains infinitely many occurrences
  of~$1$, the word $v_1$ belongs to $\{0,1\}^*10^*$. Let $u\in U$ be
  arbitrary (such a word exists since $UV^\omega\neq\emptyset$) and
  consider the $\omega$-words $x'=u(v_1v_2)^\omega$ and
  $y'=u(v_1v_1)^\omega$ from $UV^\omega\subseteq H$. Then
  $x'\not\sim_e y'$ since $v_1\neq v_2$ and $|v_1|=|v_2|$. At the same
  time, $\supp(x')\cap\supp(y')$ is infinite since $v_1$ contains an
  occurrence of~$1$. Hence $\{x',y'\}\in E_2$.

  Thus, we found $\omega$-words $x,y,x',y'\in H$ with $\{x,y\}\in E_1$
  and $\{x',y'\}\notin E_1$ proving that $H$ is not homogeneous.
\end{proof}

Thus, we found a $(2,2)$-partition $G=(V,E_1,E_2)$ with $\cont$
elements and an injective $\omega$-automatic presentation $(L,h)$ such
that
\begin{enumerate}
\item $G$ has sets $L_1$ and $L_2$ in $\cerCF$ of size $\cont$ with
  $[L_i]^2\subseteq E_i$ for $1\le i\le 2$.
\item There is no $\omega$-language $H\in\LANG^*$ with $H\subseteq L$
  such that $h(H)$ is homogeneous of size $\cont$.
\end{enumerate}
Since all context-free $\omega$-languages belong to $\LANG^*$, the
following theorem follows the same way that Lemma~\ref{L-S2} implied
Theorem~\ref{T-S2}.

\begin{theorem}\label{T-complexity}
  For all $k,\ell\ge2$ and $\lambda>\alephN$, we have
  $(\cont,\omega\mathsf{iA})\not\ramsey(\lambda,\oCF)^k_\ell$ and
  $(\cont,\omega\mathsf{iA})\not\ramsey(\lambda,\omega\mathrm{REG})^k_\ell$.
\end{theorem}

This result can be understood as another Sierpi\'nski theorem for
$\omega$-automatic $(k,\ell)$-partitions. This time, it holds for all
$k\ge2$ (not only for $k\ge3$ as Theorem~\ref{T-S2}). The price to be
paid for this is the restriction of homogeneous sets to ``simple''
ones. In particular the non-existence f regular homogeneous sets
provides a Sierpi\'nski theorem in the spirit of automatic structures.

\section*{Open questions}

Our positive result Theorem~\ref{T-uncountable} guarantees the
existence of some clique or anticlique of size~$\cont$ (and such a
clique or anticlique can even be constructed). But the following
situation is conceivable: the $\omega$-automatic graph contains large
cliques without containing large cliques that can be described by a
language from~$\cerCF$. In particular, it is not clear whether the
existence of a large clique is decidable.

A related question concerns Ramsey quantifiers. Rubin \cite{Rub08} has
shown that the set of nodes of an automatic graph whose neighbors
contain an infinite anticlique is regular (his result is much more
general, but this formulation suffices for our purpose). It is not
clear whether this also holds for $\omega$-automatic graphs. A
positive answer to this second question (assuming that it is
effective) would entail an affirmative answer to the decidability
question above.

% \bibliography{literatur} \bibliographystyle{alpha}

\end{document}